\documentclass{llncs}

\usepackage[T1]{fontenc}
\usepackage[utf8]{inputenc}

\newcommand\occ{\ensuremath{\mathsf{occ}}}
\title{Compressed Indexing with Signature Grammars}
\author{Anders Roy Christiansen \and Mikko Berggren Ettienne}
\authorrunning{A.R. Christiansen, M.B. Ettienne} 
\institute{The Technical University of Denmark}

\begin{document}

\maketitle

\begin{abstract}
    The \textit{compressed indexing problem} is to preprocess a string
        $S$ of length $n$ into a compressed representation that supports
        pattern matching queries. That is, given a string $P$ of length
        $m$ report all occurrences of $P$ in $S$. 

        We present a data structure that supports pattern matching queries in
        $O(m + \occ (\lg\lg n + \lg^\epsilon z))$ time using $O(z \lg(n / z))$
        space where $z$ is the size of the LZ77 parse of $S$ and $\epsilon > 0$ is an arbitrarily small constant,
        when the alphabet is small
        or $z = O(n^{1 - \delta})$ for any constant $\delta > 0$.
        We also present two data structures for the general case;
        one where the space is increased by $O(z\lg\lg z)$, and
        one where the query time changes from worst-case to expected.
        These results improve the previously best known solutions.
        Notably, this is the first data structure that decides if $P$ occurs in $S$
        in $O(m)$ time using $O(z\lg(n/z))$ space.

        Our results are mainly obtained by a novel combination of 
        a randomized grammar construction algorithm
        with well known techniques relating pattern matching to 2D-range reporting.
\end{abstract}


\section{Introduction}



Given a string $S$ and a pattern $P$, the core problem of pattern matching is
to report all locations where $P$ occurs in $S$. Pattern matching problems can
be divided into two: the algorithmic problem where the text and the pattern are given
at the same time, and the data structure problem where one is allowed to
preprocess the text (pattern) before a query pattern (text) is given.
Many problems within both these categories are well-studied in the history of
stringology, and optimal solutions to many variants have been found.


In the last decades, researchers have shown an increasing interest in the compressed version of this
problem, where the space used by the index is related to the size of
some compressed representation of $S$ instead of the length of $S$. This
could be measures such as the size of the LZ77-parse of $S$,
the smallest grammar representing $S$, the number of runs in the BWT of $S$, etc. see e.g. \cite{gagie2017optimal,bille2017time,gagie2014lz77,gagie2012faster,Nishimoto2015,navarro2007compressed,Karkkainen1996}.
This problem is highly relevant as the amount of highly-repetitive data increases
rapidly, and thus it is possible to handle greater amounts of data by
compressing it. The increase in such data is due to things like
DNA sequencing, version control repositories, etc.

In this paper we consider what we call the
\textit{compressed indexing problem}, which is to preprocess a string
$S$ of length $n$ into a compressed representation that supports fast
\textit{pattern matching queries}. That is, given a string $P$ of length $m$,
report all $\occ$ occurrences of substrings in $S$ that match $P$.

Table \ref{table} gives an overview of the results on this problem.

\begin{table}[]
\centering
\caption{Selection of previous results and our new results on compressed indexing. The variables are the text size $n$, the LZ77-parse size $z$, the pattern length $m$,
$\occ$ is the number of occurrences and $\sigma$ is the size of the alphabet. (The time complexity marked by \dag\ is expected whereas all others are worst-case)}
\label{table}
\begin{tabular}{lllll}
\multicolumn{1}{l|}{Index}                                   & \multicolumn{1}{l|}{Space}                                & \multicolumn{1}{l|}{Locate  time}                                                              & $\sigma$       &  \\ \cline{1-4}
\multicolumn{1}{l|}{Gagie et al. \cite{gagie2014lz77}}     & \multicolumn{1}{l|}{$O(z\lg(n/z))$}                       & \multicolumn{1}{l|}{$O(m\lg m + \occ \lg \lg n)$}                                              & $O(1)$              &  \\
\multicolumn{1}{l|}{Nishimoto et al. \cite{Nishimoto2015}} & \multicolumn{1}{l|}{$O(z\lg n\lg^* n)$}                   & \multicolumn{1}{p{100px}|}{$O(m \lg \lg n \lg \lg z +\lg z \lg m \lg n ( \lg^* n)^2 +  \occ \lg n)$} & $n^{O(1)}$          &  \\
\multicolumn{1}{l|}{Bille et al. \cite{bille2017time}}     & \multicolumn{1}{l|}{$O(z(\lg(n/z) + \lg^\epsilon z))$}                       & \multicolumn{1}{l|}{$O(m + \occ(\lg^\epsilon n + \lg \lg n))$}                                 & $n^{O(1)}$              &  \\
\multicolumn{1}{l|}{Bille et al. \cite{bille2017time}}     & \multicolumn{1}{l|}{$O(z\lg(n/z)\lg \lg z)$}              & \multicolumn{1}{l|}{$O(m + \occ \lg \lg n)$}                                                   & $O(1)$          &  \\
\multicolumn{1}{l|}{Bille et al. \cite{bille2017time}}     & \multicolumn{1}{l|}{$O(z\lg(n/z))$}              & \multicolumn{1}{p{100px}|}{$O(m(1 + \frac{\lg^\epsilon z}{\lg(n/z)}) + \occ(\lg^\epsilon n + \lg \lg n))$}                                                   & $O(1)$          &  \\
\multicolumn{1}{l|}{\textbf{Theorem 1}}            & \multicolumn{1}{l|}{\textbf{$O(z\lg(n/z))$}}              & \multicolumn{1}{l|}{\textbf{$O(m + \occ (\lg^\epsilon z + \lg \lg n))$}}                       & \textbf{$O(1)$}     &  \\
\multicolumn{1}{l|}{\textbf{Theorem 2 (1)}}        & \multicolumn{1}{l|}{\textbf{$O(z(\lg(n/z) + \lg \lg z))$}} & \multicolumn{1}{l|}{\textbf{$O(m + \occ (\lg^\epsilon z + \lg \lg n))$}}                       & \textbf{$n^{O(1)}$} &  \\
\multicolumn{1}{l|}{\textbf{Theorem 2 (2)}}                             & \multicolumn{1}{l|}{\textbf{$O(z(\lg(n/z))$}}                                  & \multicolumn{1}{l|}{\textbf{$O(m + \occ (\lg^\epsilon z + \lg \lg n))^\dag$}}                                          & \textbf{$n^{O(1)}$} &  \\
\end{tabular}
\end{table}



\subsection{Our Results}

In this paper we improve previous solutions that are bounded by the size of the LZ77-parse.
For constant-sized alphabets we obtain the following result:

\begin{theorem}\label{thm:main}
    Given a string $S$ of length $n$ from a constant-sized alphabet with an LZ77
    parse of length $z$, we can build a compressed-index supporting pattern matching queries in
    $O(m + \occ( \lg \lg n + \lg^{\epsilon} z))$ time using $O(z \lg(n/z))$ space.
\end{theorem}

\noindent In particular, we are the first to obtain optimal search time using only $O(z\lg(n/z))$ space. For general alphabets we obtain the following:

\begin{theorem}\label{thm:general}
    Given a string $S$ of length $n$ from an integer alphabet polynomially bounded by $n$
    with an LZ77-parse of length $z$, we can build a compressed-index supporting pattern matching queries in:
    \begin{enumerate}
        \item[(1)] $O(m + \occ( \lg \lg n + \lg^{\epsilon} z))$ time using $O(z (\lg(n/z) + \lg\lg z))$ space.
        \item[(2)] $O(m + \occ( \lg \lg n + \lg^{\epsilon} z))$ expected time using $O(z \lg(n/z))$ space.
        \item[(3)] $O(m + \lg^\epsilon z + \occ( \lg \lg n + \lg^{\epsilon} z))$ time using $O(z \lg(n/z))$ space.
    \end{enumerate}
\end{theorem}
Note $\lg\lg z = O(\lg(n/z))$ when either the alphabet size is  $O(2^{\lg^{\epsilon} n})$ or
$z = o(\frac{n}{ \lg^{\epsilon' n}})$ where $\epsilon$ and $\epsilon'$  are
arbitrarily small positive constants.
Theorem~\ref{thm:main} follows directly from Theorem~\ref{thm:general} (1) given these observations.
Theorem~\ref{thm:general} is a consequence of Lemma~\ref{lem:longstrings},~\ref{lem:shortstrings},~\ref{lem:semishort}~and~\ref{lem:expectedsolution}.


\subsection{Technical Overview}


Our main new contribution is based on a new grammar construction.
In \cite{mehlhorn1997maintaining} Melhorn et al. presented a way to maintain
dynamic sequences subject to equality testing using a technique called
signatures. They presented two signature construction techniques.
One is randomized and leads to complexities that hold in expectation.
The other is based on a deterministic coin-tossing technique of Cole and Vishkin \cite{COLE198632}
and leads to worst-case running times but incurs an iterated logarithmic overhead compared to the randomized solution.
This technique has also resembles the string labeling techniques found e.g. in \cite{548491}.
To the best of our knowledge, we are the first to consider grammar compression based on the randomized solution from \cite{mehlhorn1997maintaining}. Despite it being randomized we show how to obtain worst-case query bounds for text indexing using this technique.

The main idea in this grammar construction is that similar substrings will be
parsed almost identically. This property also holds true for the deterministic
construction technique which has been used to solve dynamic string problems with and without compression,
see e.g. \cite{Nishimoto2015,alstrup2000pattern}. 
In \cite{jez2015faster} Je{\.z} devices a different grammar construction algorithm with
similar properties to solve the algorithmic pattern matching problem on grammar compressed strings
which has later been used for both static and dynamic string problems, see \cite{tomohiro2016longest,gawrychowski2015optimal}





Our primary solution has an $\lg^\epsilon z$ term in the query time which is problematic for short query patterns. To handle this, we show different solutions for handling short query patterns. These are based on the techniques from LZ77-based indexing combined with extra data structures to speed up the queries.

\section{Preliminaries}

We assume a standard unit-cost RAM model with word size $\Theta(\lg n)$ and that the input is from an integer alphabet
$\Sigma = \{1,2,\ldots, n^{O(1)}\}$. We measure space complexity in terms of machine words unless explicitly stated otherwise.
A string $S$ of length $n = |S|$ is a sequence of $n$ symbols $S[1]\ldots S[n]$ drawn from an alphabet $\Sigma$.
The sequence $S[i,j]$ is the \textit{substring} of $S$ given by $S[i]\ldots S[j]$
and strings can be concatenated, i.e. $S = S[1, k]S[k+1, n]$.
The empty string is denoted $\epsilon$ and $S[i,i] = S[i]$ while $S[i,j] = \epsilon$ if $j < i$, $S[i,j] = S[1,j]$ if $i < 1$
and $S[i,n]$ if $j > n$. The reverse of $S$ denoted $rev(s)$ is the string $S[n]S[n-1]\ldots S[1]$.
A \textit{run} in a string $S$ is a substring $S[i,j]$ with identical letters, i.e. $S[k] = S[k+1]$ for $k = i,\ldots, j-1$.
Let $S[i,j]$ be a run in $S$ then it is a \textit{maximal run} if it cannot be extended, i.e. $S[i-1] \neq S[i]$ and $S[j] \neq S[j+1]$.
If there are no runs in $S$ we say that $S$ is \textit{run-free} and it follows that $S[i] \neq S[i+1]$ for $1 \leq i < n$.
Denote by $[u]$ the set of integers $\{1,2, \ldots, u\}$.


Let $X \subseteq [u]^2$ be a set of points in a 2-dimensional grid. The \textit{2D-orthogonal range reporting problem}
is to  compactly represent $Z$ while supporting \textit{range reporting queries}, that is, given a rectangle
$R = [a_1, b_1] \times [a_2, b_2]$ report all points in the set $R \cap X$. We use the following:
\begin{lemma}[Chan et al. \cite{Larsen11}]
    For any set of $n$ points in $[u] \times [u]$ and constant $\epsilon > 0$,
    we can solve 2D-orthogonal range reporting with $O(n\lg n)$ expected preprocessing
    time using:
    \begin{enumerate}
        \item[i] $O(n)$ space and $(1 + k) \cdot O(\lg^\epsilon n \lg \lg u)$ query time
        \item[ii] $O(n\lg \lg n)$ space and $(1 + k) \cdot O(\lg \lg u)$ query time
    \end{enumerate}
    where $k$ is the number of occurrences inside the rectangle.
    \label{lem:range}
\end{lemma}

A \textit{Karp-Rabin fingerprinting function} \cite{Karp:1987} is a randomized hash function for strings.
Given a string $S$ of length $n$ and a fingerprinting function $\phi$ we can in $O(n)$ time and space compute and store $O(n)$ fingerprints such
that the fingerprint of any substring of $S$ can be computed in constant time. Identical strings have identical fingerprints.
The fingerprints of two strings $S$ and $S'$ \textit{collide} when $S \neq S'$ and $\phi(S) = \phi(S')$.
A fingerprinting function is \textit{collision-free} for a set of strings when there are no collisions between the fingerprints of any
two strings in the set. We can find collision-free fingerprinting function for a set of strings with total length $n$ in $O(n)$
expected time \cite{Porat}.


Let $D$ be a lexicographically sorted set of $k$ strings. The weak prefix search problem is to
compactly represent $D$ while supporting \textit{weak prefix queries}, that is, given
a query string $P$ of length $m$ report the rank of the lexicographically smallest and largest
strings in $D$ of which $P$ is a prefix. If no such strings exist, the answer can be arbitrary.

\begin{lemma}[Belazzougui et al.~\cite{BelazzouguiBPV10}, appendix H.3]
    Given a set $D$ of $k$ strings with average length $l$, from an alphabet of size $\sigma$,
    we can build a data structure using $O(k(\lg l + \lg\lg \sigma))$ bits of space
    supporting weak prefix search for a pattern $P$ of length $m$ in $O(m\lg \sigma/w + \lg m)$ time
    where $w$ is the word size.
\label{lem:fat1}
\end{lemma}

We will refer to the data structure of Lemma~\ref{lem:fat1} as a \textit{z-fast trie} following the notation from \cite{BelazzouguiBPV10}.
The $m$ term in the time complexity is due to a linear time preprocessing of the pattern and is not part
of the actual search.  
Therefore it is simple to do weak prefix search for any length $l$ substring of $P$
in $O(\lg l)$ time after preprocessing $P$ once in $O(m)$ time.

The \textit{LZ77-parse} \cite{Ziv1977} of a string $S$ of length $n$ is a string $\mathcal{Z}$ of the form
$(s_1, l_1, \alpha_1) \ldots (s_{z}, l_{z}, \alpha_{z}) \in ([n], [n], \Sigma)^z$.
We define $u_1 = 1$, $u_i = u_{i-1} + l_{i-1} + 1$ for $i > 1$.
For $\mathcal{Z}$ to be a valid parse, we require $l_1 = 0$, $s_i < u_i$, $S[u_i, u_i + l_i - 1] = S[s_i, s_i + l_i - 1]$, and $S[u_i + l_i] = \alpha_i$  for $i \in [z]$. This guarantees 
$\mathcal{Z}$ \textit{represents} $S$ and $S$ is uniquely defined in terms of $\mathcal{Z}$.
The substring $S[u_i, u_i + l_i]$ is called the $i^{th}$ phrase of the parse
and $S[s_i, s_i + l_i - 1]$ is its source.
A minimal LZ77-parse of $S$ can be found greedily in $O(n)$ time
and stored in $O(z)$ space \cite{Ziv1977}.
We call the positions $u_1 + l_1, \ldots, u_{z} + l_z$ the borders of $S$.

\section{Signature Grammars}

We consider a hierarchical representation of strings given by Melhorn~et~al. \cite{mehlhorn1997maintaining} with some slight modifications.
Let $S$ be a run-free string of length $n$ from an integer alphabet $\Sigma$
and let $\pi$ be a uniformly random permutation of $\Sigma$.
Define a position $S[i]$ as a local minimum of $S$ if
$1 < i < n$ and $\pi(S[i]) < \pi(S[i-1])$ and $\pi(S[i]) < \pi(S[i+1])$.
In the block decomposition of $S$, a block starts at position $1$ and at every local minimum in $S$ and ends just before
the next block begins (the last block ends at position $n$).
The block decomposition of a string $S$ can be used to construct the 
signature tree of $S$ denoted $sig(S)$ which is an ordered labeled tree with several useful properties.

\begin{lemma}
    Let $S$ be a run-free string $S$ of length $n$ from an alphabet $\Sigma$ and
    let $\pi$ be a uniformly random permutation of $\Sigma$ such that $\pi(c)$
    is the rank of the symbol $c \in \Sigma$ in this permutation. Then the expected length
    between two local minima in the sequence $\pi(S[1]), \pi(S[2]), \ldots,
    \pi(S[n])$  is at most 3 and the longest gap is $O(\lg n)$ in expectation.
        \label{lem:avg-length}
\end{lemma}

\begin{proof}
	First we show the expected length between two local minima is at most 3. Look at a position $1 \leq i \leq n$ in the sequence $\pi(S[1]), \pi(S[2]), \ldots, \pi(S[n])$. To determine if $\pi(S[i])$ is a local minimum, we only need to consider the two neighbouring elements $\pi(S[i - 1])$ and $\pi(S[i + 1])$ thus let us consider the triple $(\pi(S[i - 1]), \pi(S[i]), \pi(S[i + 1]))$. We need to consider the following cases. First assume $S[i - 1] \neq S[i] \neq S[i+1]$. There exist $3! = 6$ permutations of a triple with unique elements and in two of these the minimum element is in the middle. Since $\pi$ is a uniformly random permutation of $\Sigma$ all 6 permutations are equally likely, and thus there is $1/3$ chance that the element at position $i$ is a local minimum. Now instead assume $S[i - 1] = S[i+1] \neq S[i]$ in which case there is $1/2$ chance that the middle element is the smallest. Finally, in the case where $i = 1$ or $i = n$ there is also $1/2$ chance. As $S$ is run-free, these cases cover all possible cases. Thus there is at least $1/3$ chance that any position $i$ is a local minimum independently of $S$. Thus the expected number of local minima in the sequence is therefore at least $n / 3$ and the expected distance between any two local minima is at most $3$.
	
	The expected longest distance between two local minima of $O(\lg n)$ was shown in \cite{mehlhorn1997maintaining}.  
\end{proof}

\subsection{Signature Grammar Construction}
We now give the construction algorithm for the signature tree $sig(S)$.
Consider an ordered forest $F$ of trees. Initially,
$F$ consists of $n$ trees where the $i^{th}$ tree is a single node with label $S[i]$. Let the label of a tree $t$ denoted $l(t)$ be the label of its root node.
Let $l(F)$ denote the string that is given by the in-order concatenation of the labels of the trees in $F$.
The construction of $sig(S)$ proceeds as follows:

\begin{enumerate}
        \item 
            Let $t_i, \ldots, t_j$ be a maximal subrange of consecutive trees of $F$ with
            identical labels, i.e. $l(t_i) = \ldots = l(t_j)$.
            Replace each such subrange in $F$ by a new tree having as root a new node $v$ with children $t_i, \ldots, t_j$ and
            a label that identifies the number of children and their label. We call this kind of node a run node. Now $l(F)$ is run-free.
        \item Consider the block decomposition of $l(F)$.
            Let $t_i, \ldots, t_j$ be consecutive trees in $F$ such that their labels form a block in $l(F)$. 
            Replace all identical blocks $t_i, \ldots, t_j$ by a new tree having as root a new node with children $t_i, \ldots, t_j$ and a unique label. We call this kind of node a run-free node.

    \item Repeat step $1$ and $2$ until $F$ contains a single tree, we call this tree $sig(S)$.
\end{enumerate}

In each iteration the size of $F$ decreases by at least a factor of two and each iteration takes $O(|F|)$ time, thus it can be constructed in $O(n)$ time.

Consider the directed acyclic graph (DAG) of the tree $sig(S)$ where all identical subtrees are merged. Note we can
store run nodes in $O(1)$ space since all out-going edges are pointing to the
same node, so we store the number of edges along with a single edge instead of
explicitly storing each of them. For run-free nodes we use space
proportional to their out-degrees. We call this the signature DAG of $S$ denoted $dag(S)$.
There is a one-to-one correspondence between this DAG and an acyclic run-length grammar producing $S$
where each node corresponds to a production and each leaf to a terminal.

\subsection{Properties of the Signature Grammar}

\noindent We now show some properties of $sig(S)$ and $dag(S)$ that we will need later.
Let $str(v)$ denote the substring of $S$ given by the labels of the leaves of the subtree of $sig(S)$ induced by the node $v$ in left to right order.

\begin{lemma}\label{siggrammar}
    Let $v$ be a node in the signature tree for a string $S$ of length $n$.
    If $v$ has height $h$ then $|str(v)|$ is at least $2^h$ and thus $sig(S)$ (and $dag(S)$)
    has height $O(\lg n)$.
\end{lemma}

\begin{proof}
	This follows directly from the out-degree of all nodes being at least 2.
\end{proof}

Denote by $T(i, j)$ the set of nodes in $sig(S)$ that are ancestors of the $i^{th}$ through $j^{th}$ leaf of $sig(S)$.
These nodes form a sequence of adjacent nodes at every level of $sig(S)$ and we call them \textit{relevant nodes}
for the substring $S[i,j]$.


\begin{lemma}
$T(i, j)$ and $T(i', j')$ have identical nodes except at most the two first and two last nodes on each level whenever $S[i, j] = S[i', j']$.
\label{lem:size}
\end{lemma}

\begin{proof}
        Trivially, the leaves of $T(i, j)$ and $T(i', j')$ are identical if
        $S[i, j] = S[i', j']$. Now we show it is true for nodes on level $l$
        assuming it is true for nodes on level $l - 1$. We only consider the
        left part of each level as the argument for the right part is (almost) symmetric. Let $v_1, v_2, v_3, \ldots$
        be the nodes on level $l - 1$ in $T(i, j)$ and $u_1, u_2, u_3, \ldots$ 
        the nodes on level $l - 1$ in $T(i', j')$ in left to right order. From the
        assumption, we have $v_a, v_{a+1}, \ldots$ are identical with $u_b,
        u_{b+1}, \ldots$ for some $1 \leq a, b \leq 3$. When constructing the
        $l^{th}$ level of $sig(S)$, these nodes are divided into blocks.
        Let $v_{a + k}$ be the first block that starts after $v_a$
        then by the block decomposition, the first block after $u_b$ starts at $u_{b+k}$.
        The nodes  $v_1, \ldots, v_{a + k}$ are spanned by at most two blocks 
        and similarly for $u_1, \ldots, u_{b + k}$. These blocks become
        the first one or two nodes on level $l$ in $T(i, j)$ and $T(i', j')$
        respectively. The block starting at $v_{a+k}$
        is identical to the block starting at $u_{b+k}$ and the same holds for the following blocks.
        These blocks result in identical nodes on level $l$. Thus, if we
        ignore the at most two first (and last) nodes on level $l$ the remaining nodes are
        identical.
\end{proof}

We call nodes of $T(i, j)$ consistent in respect to $T(i, j)$ if they are guaranteed to be in any other $T(i', j')$ where $S[i, j] = S[i', j']$. We denote the remaining nodes of $T(i, j)$ as inconsistent. From the above lemma, it follows at most the left-most and right-most two nodes on each level of $T(i, j)$ can be inconsistent.

\begin{lemma}\label{lem:expectedsize}
    The expected size of the signature DAG $dag(S)$ is $O(z\lg(n/z))$.
\end{lemma}

\begin{proof}
We first bound the number of unique nodes in $sig(S)$ in terms of the
LZ77-parse of $S$ which has size $z$.
Consider the decomposition of $S$ into the $2z$ 
substrings $S[u_1, u_1 + l_1], S[u_1 + l_1 + 1], \ldots, S[u_z, u_z + l_z], S[u_z + l_z + 1]$
given by the phrases and borders of the LZ77-parse of $S$
and the corresponding sets of relevant nodes $R = \{T(u_1, u_1 + l_1), T(u_1 + l_1 + 1, u_1 + l_1 + 1), \ldots \}$.
Clearly, the union of these sets are all the nodes of $sig(S)$.
Since identical nodes are represented only once in $dag(S)$
we need only count one of their occurrences in $sig(S)$.
We first count the nodes at levels lower than $\lg(n/z)$.
A set $T(i, i)$ of nodes relevant to a substring of length one has no more than $O(\lg(n/z))$ such nodes.
By Lemma~\ref{lem:size} only $O(\lg(n/z))$ of the relevant nodes for a phrase
are not guaranteed to also appear in the relevant nodes of its source.
Thus we count a total of $O(z\lg(n/z))$ nodes for the $O(z)$ sets of relevant nodes.
Consider the leftmost appearance of a node appearing one or more times in $sig(S)$.
By definition, and because every node of $sig(S)$ is in at least one relevant set,
it must already be counted towards one of the sets.
Thus there are $O(z\lg(n/z))$ unique vertices in $sig(S)$ at levels lower than $\lg(n/z)$.
Now for the remaining at most $\lg(z)$ levels, there are no more than $O(z)$
nodes because the out-degree of every node is at least two.
Thus we have proved that there are $O(z \lg(n/z))$ unique nodes in $sig(S)$.
By Lemma~\ref{lem:avg-length} the average block size and thus the expected out-degree of a node is $O(1)$.
It follows that the expected number of edges and the expected size of $dag(S)$ is $O(z\lg(n/z))$.

\end{proof}

\begin{lemma}
    A signature grammar of $S$ using $O(z\lg(n/z))$ (worst case) space can be constructed in $O(n)$ expected time.
\end{lemma}

\begin{proof}
	Construct a signature grammar for $S$ using the signature grammar construction algorithm. If the average out-degree of the run-free nodes in $dag(S)$ is more than some constant greater than 3 then try again. In expectation it only takes a constant number of retries before this is not the case.
\end{proof}

\begin{lemma}
	Given a node $v \in dag(S)$, the child that produces the character at position $i$ in $str(v)$ can be found in $O(1)$ time.
\end{lemma}

\begin{proof}
	First assume $v$ is a run-free node. If we store $|str(u)|$ for each child $u$ of $v$ in order, the correct child corresponding to position $i$ can simply be found by iterating over these. However, this may take $O(\log n)$ time since this is the maximum out-degree of a node in $dag(S)$. This can be improved to $O(\log \log n)$ by doing a binary search, but instead we use a Fusion Tree from \cite{Fredman1990} that allows us to do this in $O(1)$ time since we have at most $O(\log n)$ elements. This does not increase the space usage.
	If $v$ is a run node then it is easy to calculate the right child by a single division.
\end{proof}


\section{Long Patterns}

In this section we present how to use the signature grammar to construct a compressed index
that we will use for patterns of length $\Omega(\lg^\epsilon z)$ for constant $\epsilon > 0$. We obtain the following lemma:

\begin{lemma}\label{lem:longstrings}
    Given a string $S$ of length $n$ with an LZ77-parse of length $z$
    we can build a compressed index supporting pattern matching queries in
    $O(m + (1 + \occ)\lg^\epsilon z)$ time using $O(z\lg(n/z))$ space for any constant $\epsilon > 0$.
\end{lemma}

\subsection{Data Structure}
Consider a vertex $v$ with children $u_1, \ldots u_k$ in $dag(S)$.
Let $pre(v, i)$ denote the prefix of $str(v)$ given by
concatenating the strings represented by the first $i$ children of $v$
and let $suf(v, i)$ be the suffix of $str(v)$ given by
concatenating the strings represented by the last $k - i$ children of $x$.


The data structure is composed of two z-fast tries (see Lemma~\ref{lem:fat1}) $T_1$ and $T_2$ and a 2D-range reporting
data structure $R$.

For every non-leaf node $v \in dag(S)$ we store the following. Let $k$ be the number of children of $v$ if $v$ is a run-free node otherwise let $k = 2$:

	\begin{itemize}
	    \item The reverse of the strings $pre(v, i)$ for $i \in [k - 1]$ in the z-fast trie $T_1$.
	    \item The strings $suf(v, i)$ for $i \in [k - 1]$ in the z-fast trie $T_2$.
	    \item The points $(a, b)$ where
	    $a$ is the rank of the reverse of $pre(v, i)$ in $T_1$ and $b$ is the rank of $suf(v, i)$ in $T_2$ for $i \in [k - 1]$
        are stored in $R$. A point stores the vertex $v \in dag(S)$ and the length of $pre(v, i)$ as auxiliary information.
    \end{itemize}

    There are $O(z\lg(n/z))$ vertices in $dag(S)$ thus $T_1$ and $T_2$ take no more than $O(z\lg(n/z) )$ words of space using Lemma~\ref{lem:fat1}.
    There $O(z\lg(n/z))$ points in $R$ which takes $O(z\lg(n/z))$ space using Lemma~\ref{lem:range} (i) thus the total space in words is $O(z\lg(n/z))$.

\subsection{Searching}

Assume in the following that there are no fingerprint collisions. Compute all
the prefix fingerprints of $P$ $\phi(P[1]), \phi(P[1, 2]), \ldots, \phi(P[1,
m])$. Consider the signature tree $sig(P)$ for $P$.
Let $l_i^k$ denote the $k$'th left-most vertex on level $i$ in $sig(P)$ and let $j$ be the last level.
Let $P_L = \{ |str(l_1^1)|, |str(l_1^1)| + |str(l_1^2)|, |str(l_2^1)|,
|str(l_2^1)| + |str(l_2^2)|, \ldots, |str(l_j^1)|, |str(l_j^1)| + |str(l_j^2)| \}$.
Symmetrically, let $r_i^k$ denote the $k$'th right-most vertex on level $i$ in $sig(P)$
and let $P_R = \{ m - |str(r_1^1)|, m - |str(r_1^1)| - |str(r_1^2)|, m - |str(r_2^1)|, m - |str(r_2^1)| - |str(r_2^2)|, \ldots, m - |str(r_j^1)|, m - |str(r_j^1)| - |str(r_j^2)| \}$.
Let $P_{S} = P_L \cup P_R$.


For $p \in P_S$ search for the reverse of $P[1, p]$ in $T_1$
and for $P[p+1, m]$ in $T_2$ using the precomputed fingerprints. Let $[a,b]$ and $[c,d]$ be the respective
ranges returned by the search. Do a range reporting query for the
(possibly empty) range $[a,b] \times [c,d]$ in $R$.
Each point in the range identifies a node $v$
and a position $i$ such that $P$ occurs at position
$i$ in the string $str(v)$.
If $v$ is a run node, there is furthermore an occurrence of $P$ in $str(v)$ for all positions $i + k\cdot |str(child(v))|$ where $k = 1, \ldots, j$
and $j \cdot |str(child(v))| + m \leq str(v)$.

To report the actual occurrences of $P$ in $S$ we traverse all ancestors of $v$ in $dag(S)$; for each occurrence of $P$ in $str(v)$ found, recursively visit each parent $u$ of $v$ and offset the location of the occurrence to match the location in $str(u)$ instead of $str(v)$. When $u$ is the root, report the occurrence.
Observe that the time it takes to traverse the ancestors of $v$ is linear in the number of occurrences we find.

We now describe how to handle fingerprint collisions.
Given a z-fast trie, Gagie et al. \cite{gagie2014lz77} show how to perform
$k$ weak prefix queries and identify all false positives
using $O(k\lg m + m)$ extra time by
employing bookmarked extraction and bookmarked fingerprinting.
Because we only compute fingerprints and extract prefixes (suffixes)
of the strings represented by vertices in $dag(S)$ we do not need bookmarking to do this.
We refer the reader to \cite{gagie2014lz77} for the details.
Thus, we modify the search algorithm such that all the searches in $T_1$ and $T_2$ are carried out first,
then we verify the results before progressing to doing range reporting queries only
for ranges that were not discarded during verification.

\subsection{Correctness}

For any occurrence $S[l, r]$ of $P$ in $S$ there is a node $v$ in $sig(S)$ that stabs $S[l, r]$, ie. a suffix of $pre(v, i)$ equals a prefix $P[1, j]$ and a prefix of $suf(v, i)$ equals the remaining suffix $P[j + 1, m]$ for some $i$ and $j$. Since we put all combinations of $pre(v, i)$, $suf(v, i)$ into $T_1, T_2$ and $R$, we would be guaranteed to find all nodes $v$ that contains $P$ in $str(v)$ if we searched for all possible split-points $1, \ldots, m-1$ of $P$
i.e. $P[1, i]$ and $P[i + 1, m]$ for $i = 1, \ldots, m-1$.

We now argue that we do not need to search for all possible split-points of $P$ but only need to consider those in the set $P_S$. For a position $i$, we say the node $v$ stabs $i$ if the nearest common ancestor of the $i^{th}$ and $i+1^{th}$ leaf of $sig(S)$ denoted $NCA(l_i, l_{i+1})$ is $v$.

Look at any occurrence $S[l, r]$ of $P$. Consider $T_S = T(l, r)$ and $T_P = sig(P)$. Look at a possible split-point $i \in [1, m - 1]$ and the node $v$ that stabs position $i$ in $T_P$. Let $u_l$ and $u_r$ be adjacent children of $v$ such that the rightmost leaf descendant of $u_l$ is the $i^{th}$ leaf and
the leftmost leaf descendant of $u_r$ is the $i+1^{th}$ leaf.
We now look at two cases for $v$ and argue it is irrelevant to consider position $i$ as split-point for $P$ in these cases:

\begin{enumerate}
	\item \textbf{Case $v$ is consistent (in respect to $T_P$)}. In this case it is guaranteed that the node that stabs $l + i$ in $T_S$ is identical to $v$. Since $v$ is a descendant of the root of $T_P$ (as the root of $T_P$ is inconsistent) $str(v)$ cannot contain $P$ and thus it is irrelevant to consider $i$ as a split-point.
	\item \textbf{Case $v$ is inconsistent and $u_l$ and $u_r$ are both consistent (in respect to $T_P$)}. In this case $u_l$ and $u_r$ have identical corresponding nodes $u_l'$ and $u_r'$ in $T_S$. Because $u_l$ and $u_r$ are children of the same node it follows that $u_l'$ and $u_r'$ must also both be children of some node $v'$ that stabs $l + i$ in $T_S$ (however $v$ and $v'$ may not be identical since $v$ is inconsistent). Consider the node $u_{ll}'$ to the left of $u_l'$ (or symmetrically for the right side if $v$ is an inconsistent node in the right side of $T_P$). If $str(v')$ contains $P$ then $u_{ll}'$ is also a child of $v'$ (otherwise $u_{l}$ would be inconsistent). So it suffices to check the split-point $i - |u_l|$. Surely $i - |u_l|$ stabs an inconsistent node in $T_P$, so either we consider that position relevant, or the same argument applies again and a split-point further to the left is eventually considered relevant.
\end{enumerate}

Thus only split-points where $v$ and at least one of $u_l$ or $u_r$ are inconsistent are relevant. These positions are a subset of the position in $P_S$, and thus we try all relevant split-points.
 



\subsection{Complexity}

A query on $T_1$ and $T_2$ takes $O(\lg m)$ time by Lemma~\ref{lem:fat1} while a query on $R$ takes $O(\lg^\epsilon z)$ time using Lemma~\ref{lem:range} (i) (excluding reporting).
We do $O(\lg m)$ queries as the size of $P_S$ is $O(\lg m)$. Verification of the $O(\lg m)$ strings we search for
takes total time $O(\lg^2 m + m) = O(m)$. Constructing
the signature DAG for $P$ takes $O(m)$ time, thus total time without reporting is
$O(m + \lg m \lg^\epsilon z) = O(m + \lg^{\epsilon'} z)$ for any $\epsilon' > \epsilon$.
This holds because if $m \leq \lg^{2\epsilon} z$ then $\lg m \lg^\epsilon z \leq \lg \lg^{2\epsilon} z \lg^\epsilon z = O(\lg^{\epsilon'} z)$, otherwise $m > \lg^{2\epsilon} z \Leftrightarrow \sqrt{m} > \lg^\epsilon z$ and then $\lg m \lg^\epsilon z = O(\lg m \sqrt{m}) = O(m)$. For every query on $R$ we may find multiple points each corresponding to an occurrence of $P$.
It takes $O(\lg^\epsilon  z)$ time to report each point thus the total time becomes $O(m + (1 + \occ)\lg^{\epsilon'} z)$.


\section{Short Patterns}

Our solution for short patterns uses properties of the LZ77-parse of $S$.
A \textit{primary} substring of $S$ is a substring that contains one or more borders of $S$, all other substrings are called $\textit{secondary}$.
A primary substring that matches a query pattern $P$ is a \textit{primary occurrence} of $P$
while a secondary substring that matches $P$ is a \textit{secondary occurrence} of $P$.
In a seminal paper on LZ77 based indexing \cite{Karkkainen1996} K{\"{a}}rkk{\"{a}}inen and Ukkonen
use some observations by Farach and Thorup \cite{Farach1998} to
show how all secondary occurrences of a query pattern $P$ can be found
given a list of the primary occurrences of $P$
through a reduction to orthogonal range reporting. 
Employing the range reporting result given in Lemma~\ref{lem:range} (ii),
all secondary occurrences can be reported as stated in the following lemma: 

\begin{lemma}[K{\"{a}}rkk{\"{a}}inen and Ukkonen \cite{Karkkainen1996}]\label{lem:lz77secondary}
	Given the LZ77-parse of a string $S$ there exists a data structure that uses $O(z \lg \lg z)$ space that can report all secondary occurrences of a pattern $P$ given the list of primary occurrences of $P$ in $S$ in $O(\occ \lg \lg n)$ time.
\end{lemma}

We now describe a data structure that can report all primary occurrences of a pattern $P$ of length at most $k$ in $O(m + \occ)$ time using $O(zk)$ space. 

\begin{lemma}\label{lem:shortstrings}
	Given a string $S$ of length $n$ and a positive integer $k \leq n$
	we can build a compressed index supporting pattern matching queries for patterns of length $m$
	in $O(m + \occ\lg\lg n)$ time using $O(zk + z\lg\lg z)$ space that works for $m \leq k$.
\end{lemma}

\begin{proof}
Consider the set $C$ of $z$ substrings of $S$ that are defined by $S[u_i-k, u_i+k - 1]$ for $i \in [z]$, ie. the substrings of length $2k$ surrounding the borders of the LZ77-parse. The total length of these strings is $\Theta(zk)$. 
Construct the generalized suffix tree $T$ over the set of strings $C$. This takes $\Theta(zk)$ words of space. To ensure no occurrence is reported more than once, if multiple suffixes in this generalized suffix tree correspond to substrings of $S$ that starts on the same position in $S$, only include the longest of these. This happens when the distance between two borders is less than $2k$.

To find the primary occurrences of $P$ of length $m$, simply find all occurrences of $P$ in $T$. These occurrences are a super set of the primary occurrences of $P$ in $S$, since $T$ contains all substrings starting/ending at most $k$ positions from a border. It is easy to filter out all occurrences that are not primary, simply by calculating if they cross a border or not. This takes $O(m + \occ)$ time (where $\occ$ includes secondary occurrences). Combined with Lemma~\ref{lem:lz77secondary} this gives Lemma~\ref{lem:shortstrings}.
\end{proof}

\section{Semi-Short Patterns}

In this section, we show how to handle patterns of length between $\lg \lg z$ and $\lg^\epsilon z$. It is based on the same reduction to 2D-range reporting as used for long patterns. However, the positions in $S$ that are inserted in the range reporting structure is now based on the LZ77-parse of $S$ instead. Furthermore we use Lemma~\ref{lem:range} (ii) which gives faster range reporting but uses super-linear space, which is fine because we instead put fewer points into the structure. We get the following lemma:

\begin{lemma}\label{lem:semishort}
    Given a string $S$ of length $n$ we solve the compressed indexing problem
    for a pattern $P$ of length $m$ with $\lg\lg z \leq m \leq \lg^\epsilon z$ for any positive constant $\epsilon < \frac{1}{2}$ in $O(m + \occ(\lg\lg n + \lg^\epsilon z))$ time using $O(z(\lg \lg z + \log(n/z)))$ space.
\end{lemma}

\subsection{Data Structure}

As in the previous section for short patterns, we only need to worry about primary occurrences of $P$ in $S$. Let $B$ be the set of all substrings of length at most $\lg^\epsilon z$ that cross a border in $S$. 
The split positions of such a string are the offsets of the leftmost borders in its occurrences.
All primary occurrences of $P$ in $S$ are in this set. The size of this set is $|B| = O(z \lg^{2\epsilon} z)$. The data structure is composed by the following:

\begin{itemize}
    \item A dictionary $H$ mapping each string in $B$ to its split positions.
    \item A z-fast trie $T_1$ on the reverse of the strings $T[u_i, l_i]$ for $i \in [z]$.
    \item A z-fast trie $T_2$ on the strings $T[u_i, n]$ for $i \in [z]$.
    \item A range reporting data structure $R$ with a point $(c,d)$ for every pair
        of strings $C_i = T[u_i,l_i], D_i = T[u_{i+1}, n]$ for $i \in [z]$ where $D_{z} = \epsilon$ and
        $c$ is the lexicographical rank of the reverse of $C_i$ in the set $\{C_1, \ldots, C_{z} \}$
        and $d$ is the lexicographical rank of $D_i$ in the set $\{D_1, \ldots D_{z}\}$. 
        We store the border $u_i$ along with the point $(c, d)$.
    \item The data structure described in Lemma~\ref{lem:lz77secondary} to report secondary occurrences.
    \item The signature grammar for $S$. 
\end{itemize}

Each entry in $H$ requires $\lg \lg^\epsilon z = O(\lg \lg z)$ bits to store since a split position can be at most $\lg^\epsilon z$. Thus the dictionary can be stored in $O(|B| \cdot \lg \lg z) = O(z \lg^{2\epsilon} z \lg \lg z)$ bits which for $\epsilon < \frac{1}{2}$ is $O(z)$ words. The tries $T_1$ and $T_2$ take $O(z)$ space while $R$ takes $O(z\lg\lg z)$ space. The signature grammar takes $O(z \log(n/z))$. Thus the total space is $O(z(\lg \lg z + \log(n/z)))$.

\subsection{Searching}

Assume a lookup for $P$ in $H$ does not give false-positives. Given a pattern $P$ compute all prefix fingerprints of $P$. Next do a lookup in $H$. If there is no match then $P$ does
not occur in $S$. Otherwise, we do the following for each of the split-points $s$ stored in $H$. First split $P$ into a left part $P_l = P[0,s-1]$ and a right part $P_r = P[s, m]$. Then search for the reverse of $P_l$ in
$T_1$ and for $P_r$ in $T_2$ using the corresponding fingerprints. The search induces a (possibly empty) range for which 
we do a range reporting query in $R$. 
Each occurrence in $R$ corresponds to a primary occurrence of $P$ in $S$, so report these. Finally use Lemma~\ref{lem:lz77secondary} to report all secondary occurrences.

Unfortunately, we cannot guarantee a lookup for $P$ in $H$ does not give a
false positive. Instead, we pause the reporting step when the first possible
occurrence of $P$ has been found. At this point, we verify the substring $P$
matches the found occurrence in $S$. We know this occurrence is around an
LZ-border in $S$ such that $P_l$ is to the left of the border and $P_r$ is to
the right of the border. Thus we can efficiently verify that $P$ actually
occurs at this position using the grammar.

\subsection{Analysis}

Computing the prefix fingerprints of $P$ takes $O(m)$ time. First, we analyze the running time in the case $P$ actually exists in $S$. The lookup in $H$ takes $O(1)$ time using perfect hashing. For each split-point we do two z-fast trie lookups in time $O(\lg m) = O(\lg \lg z)$. Since each different split-point corresponds to at least one unique occurrence, this takes at most $O(\occ \lg \lg z)$ time in total. Similarly each lookup and occurrence in the 2D-range reporting structure takes $\lg \lg z$ time, which is therefore also bounded by $O(\occ \lg \lg z)$ time. Finally, we verified one of the found occurrence against $P$ in $O(m)$ time. So the total time is $O(m + \occ \lg \lg z)$ in this case.

In the case $P$ does not exists, either the lookup in $H$ tells us that, and we spend $O(1)$ time, or the lookup in $H$ is a false-positive. In the latter case, we perform exactly two z-fast trie lookups and one range reporting query. These all take time $O(\lg \lg z)$. Since $m \geq \lg\lg z$ this is $O(m)$ time. Again, we verified the found occurrence against $P$ in $O(m)$ time. The total time in this case is therefore $O(m)$.

Note we ensure our fingerprint function is collision free for all substrings in $B$ during the preprocessing thus there can only be collisions if $P$ does not occur in $S$ when $m \leq \lg^\epsilon z$.

\section{Randomized Solution}

In this section we present a very simple way to turn the $O(m + (1 + \occ)\lg^\epsilon z)$ worst-case time of Lemma~\ref{lem:longstrings} into $O(m + \occ\lg^\epsilon z)$ expected time. First observe, this is already true if the pattern we search for occurs at least once or if $m \geq \lg^\epsilon z$.

As in the semi-short patterns section, we consider the set $B$ of substrings of $S$ of length at most $\lg^\epsilon z$ that crosses a border. Create a dictionary $H$ with $z\lg^{3\epsilon} z$ entries and insert all the strings from $B$. This means only a $\lg^\epsilon z$ fraction of the entries are used, and thus if we lookup a string $s$ (where $|s| \leq \lg^\epsilon z$) that is not in $H$ there is only a $\frac{1}{\lg^\epsilon z}$ chance of getting a false-positive. 

Now to answer a query, we first check if $m \leq \lg^\epsilon z$ in which case we look it up in $H$. If it does not exist, report that. If it does exist in $H$ or if $m > \lg^\epsilon z$ use the solution from Lemma~\ref{lem:longstrings} to answer the query.

In the case $P$ does not exist, we spend either $O(m)$ time if $H$ reports no, or $O(m + \lg^\epsilon z)$ time if $H$ reports a false-positive. Since there is only $\frac{1}{\lg^\epsilon z}$ chance of getting a false positive, the expected time in this case is $O(m)$. In all other cases, the running time is $O(m + \occ\lg^\epsilon z)$ in worst-case, so the total expected running time is $O(m + \occ\lg^\epsilon z)$.
The space usage of $H$ is $O(z\lg^{3\epsilon} z)$ bits since we only need to store one bit for each entry. This is $O(z)$ words for $\epsilon \leq 1/3$.
To sum up, we get the following lemma:

\begin{lemma}\label{lem:expectedsolution}
	Given a signature grammar for a text $S$ of length $n$ with an LZ77-parse of length $z$
	we can build a compressed index supporting pattern matching queries in
	$O(m + \occ\lg^\epsilon z)$ expected time using $O(z\lg(n/z))$ space for any constant $0 < \epsilon \leq 1/3$.
\end{lemma}

\bibliography{index}

 \end{document}